\documentclass[11pt]{article}
\pdfoutput=1 %for arxiv
\usepackage[sc]{mathpazo}

\usepackage{palatino,microtype}
\usepackage[margin=15pt,font=small,labelfont={bf}]{caption}
\usepackage[margin=1in]{geometry}
\usepackage[english]{babel}
\usepackage[utf8]{inputenc}%[utf8x]
\usepackage[compact]{titlesec}
\usepackage{cmap}
\usepackage[T1]{fontenc}
\usepackage{bm}
\usepackage{booktabs}
\usepackage{mathtools}
\usepackage{amsfonts}
\usepackage{amsmath}
\usepackage{amssymb}
\usepackage{amsthm}
\usepackage{float}
\usepackage{graphics}
\usepackage{hyperref}
\usepackage[svgnames,table]{xcolor}
\hypersetup{colorlinks={true},urlcolor={blue},linkcolor={DarkBlue},citecolor=[named]{DarkGreen},linktoc=all}
\usepackage{xspace}

\usepackage[capitalise,nameinlink,noabbrev]{cleveref}
\usepackage{doi}

\usepackage{colortbl}
\colorlet{myred}{red!25}
\colorlet{myblue}{blue!25}
\colorlet{mygreen}{green!25}

\newcommand{\BibTeX}{\rm B\kern-.05em{\sc i\kern-.025em b}\kern-.08em\TeX}

\usepackage[labelfont={normalfont,bf},textfont=it]{caption}
\usepackage{subcaption}

\usepackage{nicefrac}
\usepackage{pifont}

\usepackage{array,multirow,graphicx,bigdelim}

\input{insbox}%%%%%%%%%%%%%% TeX macro,
\usepackage{tikz}
\usepackage{tikz-dimline}
\usepackage{tikz-cd}
\usetikzlibrary{fit,calc,math,shapes,shapes.multipart,decorations.text,arrows,decorations.markings,decorations.pathmorphing,shapes.geometric,positioning,decorations.pathreplacing}
\tikzset{snake it/.style={decorate, decoration=snake}}

\colorlet{mygray}{gray!40}

\makeatletter
\let\oldnl\nl% Store \nl in \oldnl
\newcommand{\nonl}{\renewcommand{\nl}{\let\nl\oldnl}}% Remove line number for one line
\makeatother

\usepackage[capitalise,nameinlink,noabbrev]{cleveref}

\usepackage{thmtools}

\newtheorem{theorem}{Theorem}

\newtheorem{lemma}[theorem]{Lemma}

\newtheorem{definition}[theorem]{Definition}

\newcommand{\cA}{\mathcal{A}}

\newcommand{\N}{\mathbb{N}}

\newcommand{\oadj}{\text{ordered adjacent}\xspace}
\newcommand{\chainalg}{\textsc{ChainEF1}}
\newcommand{\fullalg}{\textsc{SwapEF1}}
\DeclareMathOperator{\argmax}{argmax}
\newcommand{\tI}{\tilde{I}}

\newcommand{\tM}{\tilde{M}}
\newcommand{\tE}{\tilde{E}}
\newcommand{\tG}{\tilde{G}}
\newcommand{\tv}{\tilde{v}}
\newcommand{\tA}{\tilde{A}}
\newcommand{\tcA}{\tilde{\mathcal{A}}}

\usepackage{algorithm}
\usepackage{algpseudocode}

\theoremstyle{remark}

\title{Dividing Conflicting Items Fairly}
\author{
	\begin{tabular}{m{0.12\textwidth}m{0.12\textwidth}m{0.12\textwidth}m{0.12\textwidth}%m{0.12\textwidth}m{0.12\textwidth}
 }
		\multicolumn{2}{c}{\textbf{Ayumi Igarashi}} & \multicolumn{2}{c}{\textbf{Pasin Manurangsi}} 
        \\
		\multicolumn{2}{c}{The University of Tokyo} & \multicolumn{2}{c}{Google Research}
        \\ 
		\multicolumn{2}{c}{\href{mailto:igarashi@mist.i.u-tokyo.ac.jp}{\small{\texttt{igarashi@mist.i.u-tokyo.ac.jp}}}}&
        \multicolumn{2}{c}{\href{mailto:pasin@google.com}{\small{\texttt{pasin@google.com}}}} 
        \\
        &&&\\
		\multicolumn{4}{c}{\textbf{Hirotaka Yoneda}}\\
            \multicolumn{4}{c}{The University of Tokyo}\\
            \multicolumn{4}{c}{\href{mailto: yoneda-h@g.ecc.u-tokyo.ac.jp}{\small{\texttt{yoneda-h@g.ecc.u-tokyo.ac.jp}}}}
	\end{tabular}
}

\date{}

\begin{document}

\maketitle 
%%%%%%%%%%%%%%%%%%%%%%%%%%%%%%%%%%%%%%%%%%%%%%%%%%%%%%%%%%%%%%%%%%%%%%%%
\begin{abstract}
We study the allocation of indivisible goods under conflicting constraints, represented by a graph. In this framework, vertices correspond to goods and edges correspond to conflicts between a pair of goods. Each agent is allocated an independent set in the graph. In a recent work of Kumar et al.~\cite{KEG+24}, it was shown that a maximal EF1 allocation exists for interval graphs and two agents with monotone valuations. We significantly extend this result by establishing that a maximal EF1 allocation exists for \emph{any graph} when the two agents have monotone valuations. To compute such an allocation, we present a polynomial-time algorithm for additive valuations, as well as a pseudo-polynomial time algorithm for monotone valuations. Moreover, we complement our findings by providing a counterexample demonstrating a maximal EF1 allocation may not exist for three agents with monotone valuations; further, we establish NP-hardness of determining the existence of such allocations for every fixed number $n \geq 3$ of agents. All of our results for goods also apply to the allocation of chores. 
\end{abstract}

\section{Introduction}
How can we allocate a resource fairly? This problem was first formalized by the pioneering work of Steinhaus~\cite{Steinhaus48} and has since been extensively studied in the fields of economics, mathematics, and computer science under the umbrella of \emph{fair division}. Applications of fair division arise in many real-life situations, including 
the allocation of courses among students~\cite{BudishCaKe17}, the division of family inheritance among family members~\cite{GoldmanPr14}, and the division of household chores between couples~\cite{IgarashiYo23}.

%EF and EF1, indivisible resource allocation
A central notion of fairness in fair division is \emph{envy-freeness}, which requires that every agent is allocated their most preferred bundle in the allocation. However, such a fairness guarantee is impossible to achieve when dealing with indivisible resources, such as courses, houses, or tasks. Consequently, recent literature on discrete fair division has focused on approximate fairness, exploring various concepts and algorithmic results~\cite{AAB+23fair}. One particular influential relaxation of envy-freeness is, \emph{envy-freeness up to one good (EF1)}, introduced by Budish~\cite{Budish11}, allowing agents to remove one good from others' bundle to eliminate envy.  
This concept has garnered significant attention over the past decade. It is known that for general classes of monotone valuations, an EF1 allocation exists and can be computed efficiently~\cite{LiptonMaMo04}.

%fair division under constraints, shift scheduling
Most studies on fair division assume that any allocation is feasible. While this assumption may hold in some cases, many practical scenarios involve constraints that restrict the structure of allocations. For instance, consider the allocation of multiple offices among several people over different periods of time. If the time intervals associated with two offices overlap, they cannot be assigned to the same person. Similar constraints arise in job scheduling, where overlapping shifts cannot be allocated to the same employee. Another example is the allocation of players to sports teams. If two players have overlapping areas of expertise, it is preferable not to assign them to the same team. 
A versatile framework for modeling such constraints, explored in a series of recent papers~\cite{ChiarelliKMPPS20,HummelH22,KEG+24}, represents conflicts among indivisible resources using a graph structure, where vertices correspond to resources and edges represent conflicts. 

%%%%Related work, 
%Hummel and Hetland
Conflicting constraints introduce new challenges, as standard fairness and efficiency concepts often become unattainable.
Notably, canonical efficiency concepts such as completeness and Pareto-optimality are incompatible with EF1 under these constraints. In fact, with conflicting constraints, a complete allocation that assigns all goods may not always exist. Furthermore, even if such an allocation exists, there are simple instances where a complete EF1 allocation is unattainable~\cite{HummelH22}.

Kumar et al.~\cite{KEG+24} studied chore allocation under conflicting constraints, observing that even on a path, EF1 is incompatible with Pareto-optimality for two agents with identical additive valuations.\footnote{In~\cite{KEG+24}, Pareto-optimality is defined with respect to maximal allocations. } To address this limitation, they introduced the notion of \emph{maximality}. A maximal allocation ensures that no unassigned item can be feasibly allocated to some agent. Kumar et al. showed that for interval graphs---a common structure in job scheduling---a maximal EF1 allocation among two agents always exists and can be efficiently found for monotone valuations. 
However, the existence of a maximal EF1 allocation for more general graph families remains unresolved, offering a rich avenue for further research.

\paragraph{Our contributions.}
In this paper, we study the allocation of indivisible goods under conflicting constraints. Our goal is to identify conditions under which a maximal EF1 allocation exists and can be efficiently computed. Our main contributions are as follows:

\begin{enumerate}
    \item {\bf Two-Agent-Case}: We significantly extend the result of Kumar et al.~\cite{KEG+24} by establishing that a maximal EF1 allocation exists for \emph{any graph} when the agents have monotone valuations. Further, we develop efficient algorithms for finding such allocations, including a polynomial-time algorithm for additive valuations and a pseudo-polynomial-time algorithm for monotone valuations. 
    Note that the two-agent case is of particular importance in fair division, with various applications, including inheritance division, house-chore division, and divorce settlements~\cite{BramsFi00,brams2014two,IgarashiYo23}.
    
    \item {\bf Three or More Agents}: We establish a sharp dichotomy from the two-agent case in terms of both existence and computational complexity. First, we provide an example where a maximal EF1 allocation fails to exist, even for three agents with monotone valuations. While Hummel and Hetland~\cite{HummelH22} previously identified a counterexample for four agents, no such example was known for three agents. 
    We also prove the NP-hardnesss of determining the existence of a maximal EF1 allocation for a fixed number $n\geq 3$ of agents with monotone valuations. 
   
    \item {\bf Chore Allocation}: Finally, we consider the problem of chore allocation, where each agent has a monotone non-increasing valuation. We show that the existence of a maximal EF1 allocation under identical valuations directly translates from the goods case, establishing that all our results for goods hold for chores as well. 
\end{enumerate}

\paragraph{Related work.}
%\cite{TTT06}
%conflicting items
There is a growing body of research on fair division under constraints. For a comprehensive survey on the topic, see Suksompong~\cite{Suksompong21}. Conflicting constraints in the context of fair division were introduced by Chiarelli et al.~\cite{ChiarelliKMPPS20} and have been further explored by \cite{HummelH22,KEG+24,BiswasFORC2023,LiFairScheduling2021}. Chiarelli et al.~\cite{ChiarelliKMPPS20} explored different fairness objectives from ours, focusing on partial allocations that maximize the egalitarian social welfare---defined as the value of a bundle received by the worst-off agent. Hummel and Hetland~\cite{HummelH22} studied complete allocations satisfying fairness criteria such as EF1 and MMS (maximin fair share). 
%Biswas
Biswas et al.~\cite{BiswasFORC2023} generalized the model of \cite{HummelH22}, taking into account capacity of resources.  
%Li et al. Kumar et al.
Li et al.~\cite{LiFairScheduling2021} and Kumar et al.~\cite{KEG+24} considered an interval scheduling problem, with the goal of achieving fairness concepts such as EF1 and MMS. While Li et al.~\cite{LiFairScheduling2021} focused on goods allocation with flexible intervals, Kumar et al.~\cite{KEG+24} examined chore allocation.

%equitable coloring
Our problem is closely related to the problem of \emph{equitable coloring} in graph theory, which corresponds to a complete EF1 allocation when agents have \emph{uniform} valuations (namely, all goods are equally valued $1$). In Appendix~\ref{sec:uniform}, we make this connection more precise and show that for such valuations, a maximal EF1 allocation exists on trees. 

%fair division under other constraints. Connectivity constraints
A related type of constraints to conflicting constraints is the connectivity constraints of a graph~\cite{BiloCaFl22,BouveretCeEl17}, where each agent receives a connected bundle of a graph. Note that while connectivity constraints imposed by a sparse graph such as a tree allow fewer feasible allocations, in our setting, sparsity implies greater flexibility, as it increases the number of feasible allocations.

\section{Preliminaries}\label{sec:prelim}
For any natural number $s \in \mathbb{N}$, let $[s] = \{1,2,\ldots,s\}$.

\paragraph{Problem instance.}
We use $M=[m]$ to denote the set of \emph{goods} and $N=[n]$ to denote the set of \emph{agents}. Let $G=(M,E)$ denote an undirected graph, where each vertex corresponds to a good and each edge corresponds to a conflict. Each agent $i$ has a \emph{valuation function} $v_i:2^M \rightarrow \mathbb{R}_+$; here, $\mathbb{R}_+$ is the set of non-negative reals. We assume that $v_i(\emptyset)=0$. For simplicity, the valuation of a single good $g \in M$, $v_i(\{g\})$, is denoted by $v_i(g)$. 
An instance of our problem is given by the tuple $( N, M, \mathcal{V},G)$ where $\mathcal{V} = (v_1,v_2,\dots,v_n)$ denotes a \emph{valuation profile}. We use $K_{s,t}$ to denote a complete bipartite graph with a bipartition in which one part contains $s$ vertices and another part includes $t$ vertices. 

\paragraph{Valuation function.}
A valuation function $v_i$ is \emph{monotone non-decreasing} if $v_i(S) \leq v_i(T)$ for every $S \subseteq T \subseteq M$. Unless specified otherwise, we refer to such a function simply as \emph{monotone}.  
It is \emph{additive} if $v_i(S)=\sum_{g\in S} v_i(g)$ holds for every $S \subseteq M$ and $i\in N$. A valuation profile $\mathcal{V}$ is called \emph{identical} if every agent $i \in N$ has the same valuation function; in this case, we denote this function by $v$. Let $T(m)$ denote the time to compute valuation $v_i(S)$ for a given $S \subseteq M$. 

\paragraph{Allocation.}
An \emph{allocation} is an ordered subpartition $\mathcal{A}=(A_1,\ldots, A_n)$ of $M$ where for every pair of distinct agents $i,j \in N$, $A_i \cap A_j = \emptyset$, $\cup_{i \in N} A_i \subseteq M$, and for each $i \in N$, $A_i$ is an \emph{independent set} of $G$, namely, there is no pair of goods in $A_i$ that are adjacent to each other. The subset $A_i$ is called the \emph{bundle} of agent~$i$. An allocation is \emph{complete} if all goods are allocated, i.e., $\cup_{i\in N} A_i=M$.

\paragraph{Fairness and efficiency notions.}
An allocation is \emph{envy-free} if for every pair of agents $i,j \in N$, we have $v_i(A_i) \geq v_i(A_j)$~\cite{F67resource,GS58puzzle}. It is called \emph{envy-free up to one good} (EF1) if for every pair of agents $i,j \in N$, either $A_j = \emptyset$, or there exists some good $g \in A_j$ such that $v_i(A_i) \geq v_i(A_j \setminus \{g\})$~\cite{Budish11,LiptonMaMo04}. 

As discussed in the Introduction, observe that a complete allocation may not always exist (e.g., consider a complete graph $K_2$ with one agent). Further, even when a complete allocation exists, a complete EF1 allocation may not exist: For instance, in a setting with $n$ agents and a star where the center has a value of $0$ and each of the $n+1$ leaves has a value of $1$, an agent receiving the center cannot receive any other good, while at least one agent receives two or more leaves.  

Besides completeness, another commonly used notion of efficiency in fair division is \emph{Pareto-optimality}.  Similar to the chore setting~\cite{KEG+24}, EF1 is incompatible with Pareto-optimality for goods instances. Specifically, an allocation $\mathcal{A}'$ \emph{Pareto-dominates} another allocation $\mathcal{A}$ if $v_i(A'_i) \geq v_i(A_i)$ for every $i \in N$ and the inequality is strict for at least one agent. An allocation $\mathcal{A}$ is \emph{Pareto-optimal} (PO) if there is no other allocation that Pareto-dominates it. Consider two agents and a cycle of four goods with values $1,3,1,3$ in order. In any EF1 allocation, no agent should receive both goods of value $3$, and no agent can receive two goods of different values due to conflicting constraints. However, such an allocation is Pareto-dominated by assigning the two high-value goods to one agent and the two low-value goods to the other. 

%maximality
As in Kumar et al.~\cite{KEG+24}, we therefore consider a relaxed efficiency notion of {\em maximality}. An allocation $\mathcal{A}$ is \emph{maximal} if for every agent $i \in N$ and every unallocated good $g \in M \setminus \bigcup_{i \in N}A_i$, $g$ is adjacent to some good in $A_i$. Our goal is to achieve an allocation that simultaneously satisfies maximality and EF1. 

\section{Two Agents}\label{sec:twoagents}
In this section, we consider the case of two agents.

\subsection{Cut-and-Choose Protocol}

When there are two agents, we can use a well-known strategy called \textbf{cut-and-choose protocol}~\cite{BramsTa96}.

\begin{theorem}\label{thm:cutchoose}
	Suppose that a maximal EF1 allocation always exists and can be computed in $\tau$ time for instances with two agents and identical valuations. Then, a maximal EF1 allocation always exists and can be computed in $\tau + 2T(m)$ time for instances with two agents. 
\end{theorem}

\begin{proof}
	Let $(A_1, A_2)$ a maximal EF1 allocation in a hypothetical scenario that valuations of both agents are $v_1$. Then, agent $2$ chooses a preferred bundle, leaving the reminder for $1$. The resulting allocation is maximal and EF1.
\end{proof}

Now, the question is whether a maximal EF1 allocation exists for identical valuations. In subsections 3.2-3.5, we assume that the valuations are identical and monotone, and valuations for both agents are represented by $v(S) \ (S \subseteq M)$.

%\subsection{Overall Strategy}
\subsection{Proof Strategy of Kumar et al.}

To describe our proof strategy, let us (informally) review Kumar et al.'s proof of the existence of maximal EF1 allocation when $G$ is a path. %\footnote{A graph which consists of edges $g_1 g_2, g_2 g_3, \dots, g_{m-1} g_m$.} 
At a high level, the idea of the proof is to construct a chain of maximal allocations $\cA^{(0)}, \dots, \cA^{(m-1)}$, as illustrated in Figure \ref{fig:chain_pathgraph}, satisfying the following two properties: (i) adjacent allocations only ``differ slightly'' and (ii) $(A^{(0)}_1, A^{(0)}_2) = (A^{(m-1)}_2, A^{(m-1)}_1)$. The latter implies that the signs %\footnote{For simplicity, we refer to the sign of 0 as ``positive''.} 
of $v(A^{(0)}_1) - v(A^{(0)}_2)$ and $v(A^{(m-1)}_1) - v(A^{(m-1)}_2)$ are different.
Therefore, there exists an $i$ that the signs of $v(A^{(i)}_1) - v(A^{(i)}_2)$ and $v(A^{(i + 1)}_1) - v(A^{(i + 1)}_2)$ are different. At this point, they show that at least one of $\mathcal{A}^{(i)}$ or $\mathcal{A}^{(i+1)}$ must be EF1, which shows the existence of a maximal EF1 allocation. 
This concludes the proof overview of \cite{KEG+24} for path graphs.
They also extended this method to interval graphs, although the construction of the chain of maximal allocations becomes significantly more involved. Indeed, as explained and formalized below, our main contribution is to give a construction of a chain for any graph.

\begin{figure}[htbp]
	\centering
	\includegraphics[width=0.5\textwidth]{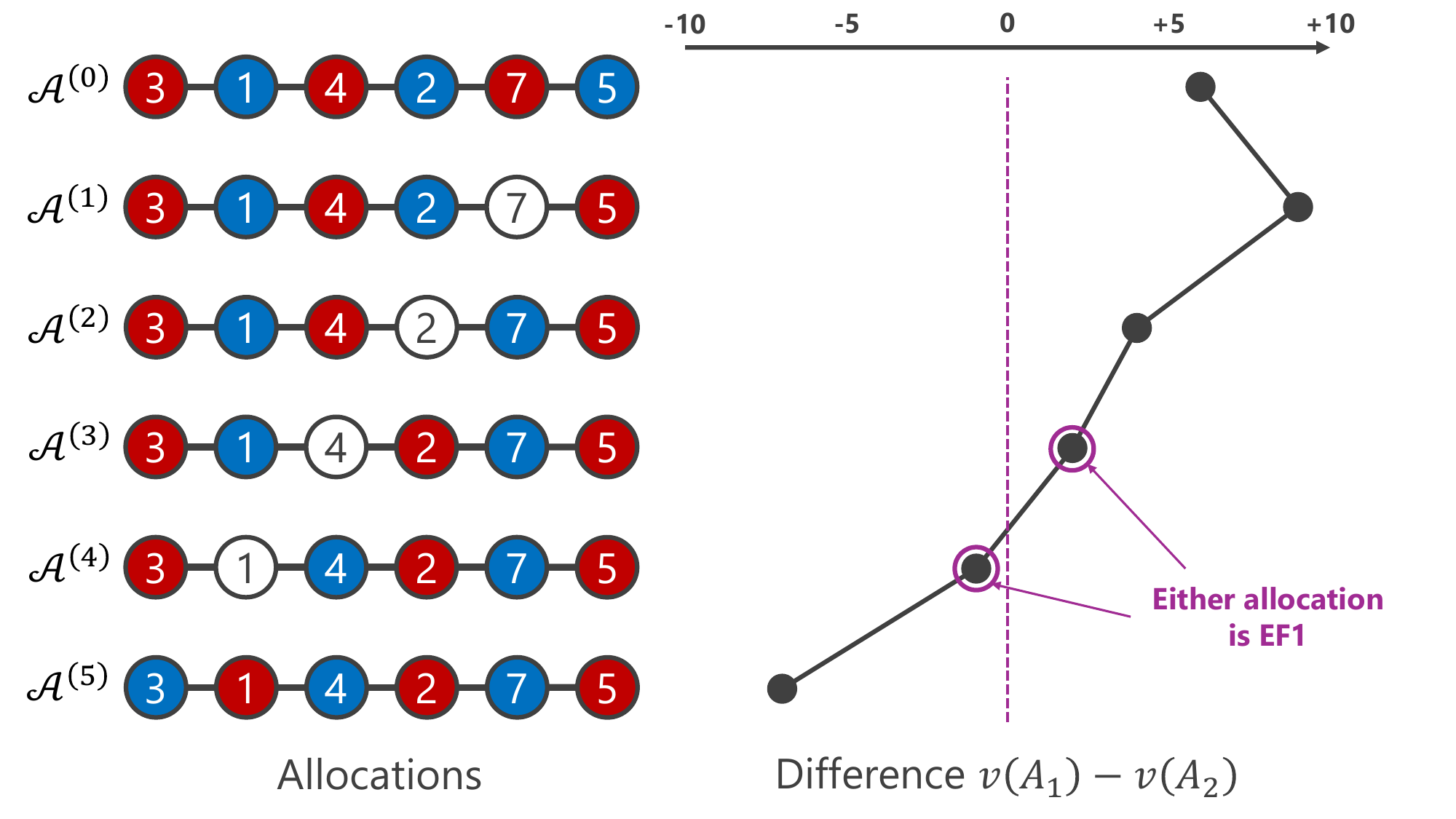}
	\caption{An example of chain of allocations for a path graph, $\mathcal{A}^{(0)}, \dots, \mathcal{A}^{(5)}$. The number written in each vertex is a valuation of the corresponding good. Red and blue vertices are those assigned to agent $1$ and $2$, respectively.}
	\label{fig:chain_pathgraph}
\end{figure}

\subsection{Useful Definitions and Lemmata}

Our construction will require several generalizations of definitions and lemmata from Kumar et al.~\cite{KEG+24}. We believe that these tools can be useful beyond the context of our work. Firstly, we use the following definition of ``adjacency''. Compared to \cite{KEG+24}, our definition (Definition~\ref{def:ordered}) is more relaxed, in that it does not place any requirement on $|A'_1 \setminus A_1|$ and $|A_2 \setminus A'_2|$ whereas Kumar et al.'s enforces that these are at most one. Such a relaxation is crucial since our ``chain'' constructed below violates the aforementioned Kumar et al.'s condition.

\begin{definition}\label{def:ordered}
A pair $(\cA, \cA')$ of allocations is \emph{\oadj} if $|A_1 \setminus A'_1| \leq 1$ and $|A'_2 \setminus A_2| \leq 1$.
\end{definition}

The following is the key lemma that enables to find an EF1 allocation at the point when $v(A_1) - v(A_2)$ crosses from positive to negative.

\begin{lemma}
Let  $(\cA, \cA')$ be any \oadj pair of allocations. Further, assume that the following conditions hold:
	\begin{enumerate}
		\item $v(A_1) \geq v(A_2)$
		\item $v(A'_1) \leq v(A'_2)$
	\end{enumerate}
	Then, at least one of $\cA$ or $\cA'$ must be EF1.
	\label{lem:chain_lemma1}
\end{lemma}

\begin{proof}
Suppose that the allocation $\cA$ is not EF1.
Since $v(A_1) \geq v(A_2)$, agent $a_2$ envies agent $a_1$ even if one good is removed from $A_1$. Since $|A_1 \setminus A'_1| \leq 1$, we must have $v(A_1 \cap A'_1) = v(A_1 \setminus (A_1 \setminus A'_1)) \geq v(A_2)$. As a result,
    \begin{align}
        v(A'_1) &\geq v(A_1 \cap A'_1) \nonumber \\ &\geq v(A_2) \nonumber \\
        &\geq v(A'_2 \cap A_2) = v(A'_2 \setminus (A'_2 \setminus A_2)), \label{eq:non-envy-oadj-claim}
    \end{align}
where %the second inequality is explained above, and 
the first and third inequalities hold because $v$ is monotone. Now, consider the allocation $\cA'$:
    \begin{itemize}
        \item Since $|A'_2 \setminus A_2| \leq 1$, \eqref{eq:non-envy-oadj-claim} implies that agent $1$ does not envy agent $2$ after removing a good from $A'_2$.
        \item Agent $2$ does not envy agent $1$ since $v(A'_1) \leq v(A'_2)$.
    \end{itemize}
    Therefore, this allocation is EF1.
\end{proof}

We generalize Kumar et al.'s method by defining a \textbf{gapless chain}, and prove that any gapless chain must contain an EF1 allocation. 
We stress again that our requirements are weaker than the chain used in Kumar et al.'s\footnote{Namely, we use a weaker adjacency notion and we only require the sign flip (first two conditions) whereas Kumar et al. require that $(A^{(0)}_1, A^{(0)}_2) = (A^{(k)}_2, A^{(k)}_1)$.}, but still suffices to ensure EF1 for identical valuations. 

\begin{definition}
A sequence of allocation $\cA^{(0)}, \dots, \cA^{(k)}$ is a \textbf{gapless chain} if it satisfies the following conditions:
	\begin{enumerate}
		\item $v(A^{(0)}_1) \geq v(A^{(0)}_2)$.
		\item $v(A^{(k)}_1) \leq v(A^{(k)}_2)$.
            \item $(\cA^{(i - 1)}, \cA^{(i)})$ is \oadj for every $i \in [k]$.
	\end{enumerate}
\end{definition}

\begin{lemma}
If $\cA^{(0)}, \dots, \cA^{(k)}$ is a gapless chain, there exists an $i \in \{0, \dots, k\}$ that $\cA^{(i)}$ is EF1.
\label{lem:chain_lemma2}
\end{lemma}

\begin{proof}
From the first two conditions, there exists $i \in [k]$ that $v(A^{(i-1)}_1) \geq v(A^{(i-1)}_2)$ and $v(A^{(i)}_1) \leq v(A^{(i)}_2)$. \Cref{lem:chain_lemma1} then implies that at least one of $\cA^{(i - 1)}$ or $\cA^{(i)}$ is EF1.
\end{proof}

%Now, the question is how to construct a gapless chain.
Given \Cref{lem:chain_lemma2}, our main task is thus to (efficiently) construct a gapless chain of maximal allocations (for any graph $G$). We devote the remainder of this section to this task.

\subsection{Proof of Existence}

\begin{algorithm}
\caption{\chainalg$(S; G = (M, E), v)$}
\label{alg:ef1-chain}
\begin{algorithmic}[1]
\Require $S = \{s_1, \dots, s_k\}$ is a maximal independent of $G$.
%\State Let $t_1, \dots, t_\ell$ be elements of $M \setminus S$
\For{$t \in M \setminus S$}
\State $\Gamma_t := \{i \in [k] \mid \{s_i, t\} \in E\}$ \vfill \Comment{Non-empty since $S$ is a maximal independent set}
\State $p_t = \min_{i \in \Gamma_t} i$
\State $q_t = \max_{i \in \Gamma_t} i$
\EndFor
\State $X_1 \gets \emptyset$
\For{$t \in M \setminus S$ in increasing order of $q_t$}
\If{$t$ has no neighbor in $X_1$}
\State $X_1 \gets X_1 \cup \{t\}$
\EndIf
\EndFor
\State $X_2 \gets \emptyset$
\For{$t \in M \setminus S$ in decreasing order of $p_t$}
\If{$t$ has no neighbor in $X_2$}
\State $X_2 \gets X_2 \cup \{t\}$
\EndIf
\EndFor
\For{$i = 0, \dots, k$}
\State $A^{(i)}_1 = \{s_{i + 1}, \dots, s_k\}  \cup \{t \in X_1 \mid q_t \leq i\}$
\State $A^{(i)}_2 =  \{s_1, \dots, s_i\} \cup \{t \in X_2 \mid p_t > i\}$
\If{$\cA^{(i)} = (A^{(i)}_1, A^{(i)}_2)$ is EF1}
\State \Return $\cA^{(i)}$
\EndIf
\EndFor
\State \Return NULL
\end{algorithmic}
\end{algorithm}

Now, we prove that, a maximal EF1 allocation always exists for two agents, as stated below.

\begin{theorem}
For $n = 2$ agents, any graph $G$ and any monotone valuation $v$, there exists a maximal EF1 allocation. \label{thm:twoagents_main}
\end{theorem}

Although we only claim the existence in the above theorem, we will in fact present an algorithm for finding such an allocation, as its running time will be discussed in the next section. Our algorithm is presented in \Cref{alg:ef1-chain} where the input $S$ can be any maximal independent set of $G$. To prove \Cref{thm:twoagents_main}, we will use one that maximizes $v(S)$. 
In Figure~\ref{fig:chain_somegraph}, we provide an example of a chain constructed by the algorithm.

\begin{figure}[htbp]
	\centering
	\includegraphics[width=0.45\textwidth]{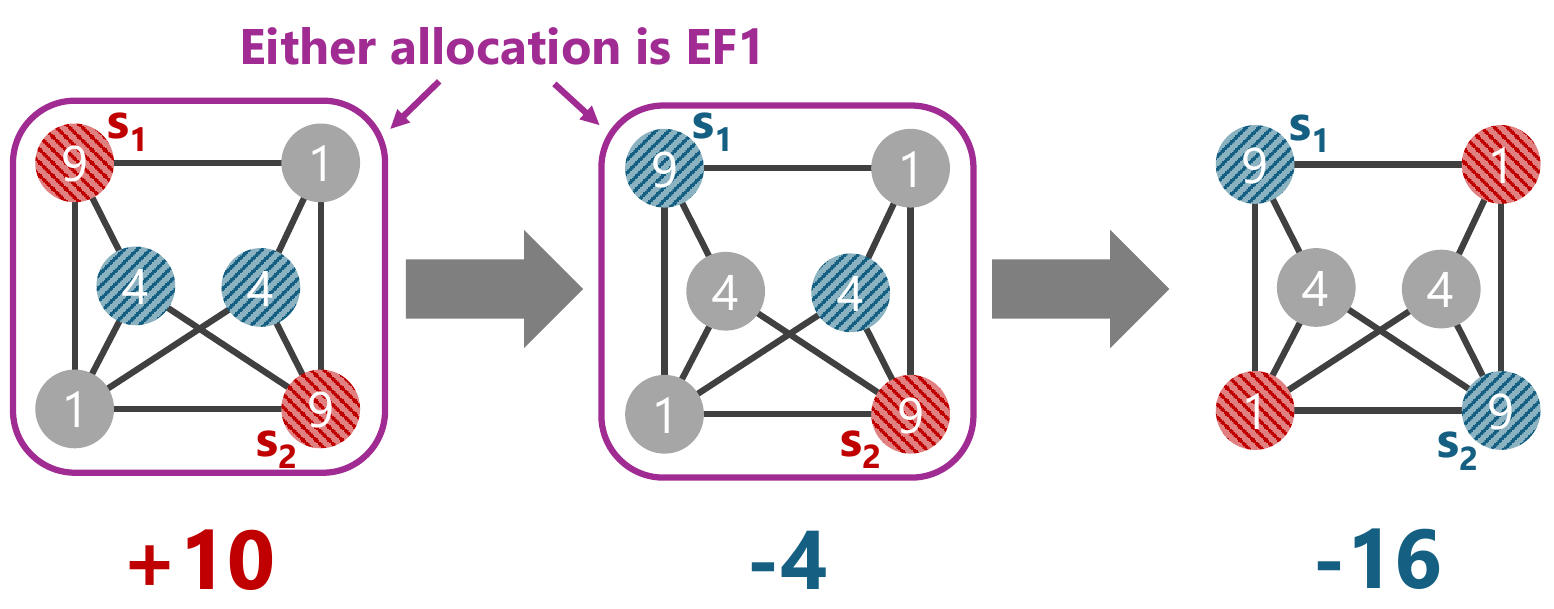}
	\caption{An example of a gapless chain of allocations, constructed by Algorithm~\ref{alg:ef1-chain}. The number written in each vertex is a valuation of the corresponding good. Red and blue vertices are those assigned to agent $1$ and $2$, respectively.
    }
	\label{fig:chain_somegraph}
\end{figure}

We now prove a couple of crucial lemmata. Starting with the fact that each allocation is valid and maximal:

\begin{lemma} \label{lem:alg-maximal}
For every $i \in \{0, \dots, k\}$, $\cA^{(i)}$ is a valid maximal allocation.
\end{lemma}

\begin{proof}
First, we prove that the allocation is valid:
\begin{itemize}
\item \textbf{No good is assigned to both agents.} This is obvious for goods in $S$. For good $t \in M \setminus S$, since $q_t \geq p_t$, at most one of the conditions $q_t \leq i$ or $p_t > i$ can hold; thus, it is assigned to at most one agent.
\item \textbf{No adjacent goods are assigned to the same agent.} Consider two goods $g, g' \ (g \neq g')$ that are assigned to agent $1$ and consider the following cases:
\begin{enumerate}
\item Both goods are in $S$. They are not adjacent because $S$ is an independent set.
\item Both goods are in $M \setminus S$. They are not adjacent because $g, g' \in X_1$ and $X_1$ is an independent set.
\item One good is from $S$ and the other is from $M \setminus S$. Assume w.l.o.g. that $g \in S$ and $g' \in M \setminus S$. From this, we must have $g \in \{s_{i + 1}, \dots, s_{k}\}$ and that $q_{g'} \leq i$. By the definition of $q_{g'}$, this ensures that $g, g'$ are not adjacent. 
\end{enumerate}
\end{itemize}

Next, we prove maximality of $\cA^{(i)}$, i.e., that no good in $M \setminus (A_1^{(i)} \cup A_2^{(i)})$ can be assigned to one of the agents. Let $g \in M \setminus (A_1^{(i)} \cup A_2^{(i)})$ be any unassigned good. Consider three following cases:
\begin{enumerate}
\item $q_g \leq i$. Since $g$ is adjacent to $s_{q_g} \in A^{(i)}_2$, $g$ cannot be assigned to agent $2$. Moreover, since $g \notin A^{(i)}_1$, it must be that $g \notin X_1$. From how $X_1$ is constructed, there exists $t \in X_1$ such that $q_t \leq q_g$ and $t$ is adjacent to $g$. As $t \in A^{(i)}_1$, $g$ cannot be assigned to agent $1$.
\item $p_g > i$. This case follows from an analogous argument to the first case.
\item $p_g \leq i < q_g$. In this case, $g$ is adjacent to $s_{p_g} \in A^{(i)}_2$ and $s_{q_g} \in A^{(i)}_1$. Thus, $t_j$ cannot be assigned. \qedhere
\end{enumerate}
\end{proof}

Next, we show that $\cA^{(0)}, \dots, \cA^{(k)}$ is a gapless chain. However, this requires an additional requirement that $v(S)$ is no smaller than $v(X_1), v(X_2)$.

\begin{lemma} \label{lem:alg-chain}
If $v(S) \geq v(X_1)$ and $v(S) \geq v(X_2)$, then $\cA^{(0)}, \dots, \cA^{(k)}$ is a gapless chain.
\end{lemma}

\begin{proof}
Note that $\cA^{(0)} = (S, X_2)$ and $\cA^{(k)} = (X_1, S)$. Thus, the first two conditions are satisfied by our assumption $v(S) \geq v(X_1), v(S) \geq v(X_2)$. Finally, for every $i \in [k]$, we can see that $A_1^{(i-1)} \setminus A_1^{(i)} = \left\{s_i\right\}$ and $A_2^{(i)} \setminus A_2^{(i-1)} = \left\{s_i\right\}$. Thus, $(\cA^{(i - 1)}, \cA^{(i)})$ is \oadj.
\end{proof}

\Cref{lem:chain_lemma2,lem:alg-maximal,lem:alg-chain} together immediately imply the following:

\begin{lemma} \label{lem:large-value-then-ef1}
If $v(S) \geq v(X_1)$ and $v(S) \geq v(X_2)$, then \Cref{alg:ef1-chain} outputs a maximal EF1 allocation.
\end{lemma}

Our main theorem of this section (\Cref{thm:twoagents_main}) now then follows easily from \Cref{lem:large-value-then-ef1} by choosing an appropriate $S$.

\begin{proof}[Proof of \Cref{thm:twoagents_main}]
Let $S$ be a maximal independent set of $G$ that maximizes $v(S)$. Since $X_1, X_2$ are independent set and $v$ is monotone, we have $v(S) \geq v(X_1)$ and $v(S) \geq v(X_2)$. Thus, \Cref{lem:large-value-then-ef1} ensures that running \Cref{alg:ef1-chain} on input $S$ yields a maximal EF1 allocation.
\end{proof}

\subsection{Algorithm}

Recall in the proof of \Cref{thm:twoagents_main} that we pick $S$ to be a maximal independent set with largest valuation. Doing so trivially would require enumerating through all $2^m$ subsets of $M$. In this section, we give a simple algorithm (\Cref{alg:swap}) that significantly improves upon this running time. In particular, it runs in polynomial-time for additive valuations and pseudo-polynomial time for general monotone valuations.

\begin{algorithm}
\caption{\fullalg$(G = (M, E), v)$}
\label{alg:swap}
\begin{algorithmic}[1]
\State $g^* \gets \argmax_{g \in M} v(g)$
\State $S^i \gets$ any maximal independent set of $G$ containing $g^*$
\State $i \gets 0$
\While{True}
\If{$\chainalg(S^i; G, v) \ne$ NULL}
\State \Return $\chainalg(S^0; G, v)$
\EndIf
\State $X^i_1, X^i_2 \gets X_1, X_2$ in $\chainalg(S^i; G, v)$
\State $\ell \gets \argmax_{\ell' \in \{1, 2\}} v(X^i_{\ell'})$
\State $S^{i+1} \gets$ any maximal independent set containing $X^i_{\ell}$
\State $i \gets i + 1$
\EndWhile
\end{algorithmic}
\end{algorithm}

Our algorithm running time is stated below in \Cref{thm:generic-alg}. Note that in the second case $B$ is the number of different values the valuation function can take. This implies that, if the valuations $v(S)$ are all integers, the running time is pseudo-polynomial.

\begin{theorem} \label{thm:generic-alg}
When there are $n = 2$ agents, there exists an algorithm that can find  a maximal EF1 allocation and its running time is as follows:
\begin{itemize}
\item $O(3^{m/3} \cdot (m \cdot T(m) + |E|))$ for monotone valuations,
\item $O(B \cdot (m \cdot T(m) + |E|))$ for monotone valuations such that the number of distinct values of $v(S)$ is at most $B$ (i.e. $|\{v(S) \mid S \subseteq M\}| \leq B$), and,
\item $O(m \log m \cdot (m \cdot T(m) + |E|))$ for additive valuations. 
\end{itemize}
\end{theorem}

\begin{proof}
Observe that the preprocessing time and the running time of each loop is $O(m \cdot T(m) + |E|)$. Thus, it suffices to bound the number of iterations of the loop in each case.
\begin{itemize}
\item For general monotone valuations, by \Cref{lem:large-value-then-ef1}, each loop either terminates or we must have $v(X^i_\ell) > v(S^i)$. From monotonicity, this implies $v(S^{i+1}) > v(S^i)$. This means that $S^0, S^1, \dots$ are all distinct maximal independent sets of $G$. Since there are at most $3^{m/3}$ maximal independent sets~\cite{Moon1965OnCI}, the number of iterations is at most $3^{m/3}$.
\item Next, suppose that $v(S)$ can take at most $B$ distinct values. From the argument above, $v(S^0), v(S^1), \dots$ are strictly increasing and are thus distinct. As a result, the number of iterations is at most $B$.
\item Finally, suppose that $v$ is additive. We claim that the following holds for each loop $i$ that does not terminate:
\begin{align} \label{eq:mult-increase-additive}
v(S^{i+1}) > \frac{m}{m - 1} \cdot v(S^i).
\end{align}
Before we prove \eqref{eq:mult-increase-additive}, let us first use it to bound the number of iterations. Notice that \eqref{eq:mult-increase-additive} implies that, after loop $i$, we must have $v(S^{i+1}) > \left(\frac{m}{m - 1}\right)^{i+1} \cdot v(S^0) \geq \left(\frac{m}{m - 1}\right)^{i+1} \cdot v(g^*)$. Moreover, our choice of $g^*$ ensures that $v(S^{i+1}) \leq m \cdot v(g^*)$. Thus, the number of iterations is at most $\log_{m/(m-1)} m + 1= O(m \log m)$.

To see that \eqref{eq:mult-increase-additive} holds, first recall from \Cref{alg:ef1-chain} that, if \Cref{alg:ef1-chain} returns NULL, it has considered either $(S^i, X^i_\ell)$ or $(X^i_\ell, S^i)$ already and has determined that this is not EF1. Since $v(X^i_\ell) > v(S^i)$, this means that, for any good $g \in X^i_\ell$, we must have $v(S) < v(X^i_\ell \setminus \left\{g\right\})$. When we pick $g \in X^i_\ell$ with the largest $v(g)$, we have
\begin{equation*}
v(g) \geq \frac{1}{\left|X^i_\ell\right|} v(X^i_\ell) \geq \frac{1}{m} v(X^i_\ell).
\end{equation*}
Hence, $v(S^i) < v(X^i_\ell \setminus \{g\}) \leq \frac{m-1}{m} v(X^i_\ell)$, proving \eqref{eq:mult-increase-additive}. \qedhere
\end{itemize}
\end{proof}

\section{Three or More Agents}\label{sec:three}

\subsection{Negative Examples}

While a maximal EF1 allocation always exists when there are two agents, it turns out that this is not the case for three agents or more.

\begin{restatable}{theorem}{ExampleThreeAgents}
 For $n = 3$ agents, there exists an instance with identical monotone valuations where no maximal EF1 allocation exists.
    \label{thm:counterexample_n3}
\end{restatable}
\begin{proof}
Consider the following instance with $7$ goods. The graph consists of a complete bipartite graph $K_{3,3}$ with bipartition $X=\{1, 2, 3\}$ and $Y=\{4, 5, 6\}$ together with two edges $\{1, 7\}$ and $\{4, 7 \}$. Each of the three agents has an identical monotone valuation $v$ such that 
	%\begin{itemize}
		%\item Number of goods: $m = 7$
		%\item Graph: all nine edges between $\{1, 2, 3\}$ and $\{4, 5, 6\}$, plus edge $(1, 7), (4, 7)$
		%\item Valuation: identical valuation that:
\begin{itemize}
			\item $v(\emptyset) = 0$; 
			\item $v(S) = 1$ if $|S|=1$ and $S\in \{ \{1\},\{4\} \}$; 
			\item $v(S) = 2$ if $|S|=1$ and $S \not \in \{ \{1\},\{4\} \}$; 
			\item $v(S) = 3$ if $S$ is $\{2, 7\}$, $\{3, 7\}$, $\{5, 7\}$, or $\{6, 7\}$; 
			\item $v(S) = 4$ for any other $S \subseteq M$. 
\end{itemize}

	Note that the symmetry on the graph and valuation holds: swapping $2$ and $3$, swapping $5$, $6$, and swapping $\{ 1, 2, 3\}$ and $\{4, 5, 6\}$ all lead to the same instance, and swapping the bundles of agents does not change whether the allocation is EF1 or not because valuations are identical. 

Specifically, consider any maximal allocation $\mathcal{A}$.
By maximality, good $7$ must be allocated to some agent since it has degree $2$. Since agents have identical valuations, without loss of generality, suppose it is allocated to agent $1$. Thus, agent $1$ can receive neither of good $1$ nor $4$; let us assume that if $1$ is allocated, then it is allocated to agent $2$ and if $4$ is allocated, then it is allocated to agent $3$. Further, agent $1$ can receive a good from at most one part of the bipartition $X$ and $Y$ due to conflicting constraints. Without loss of generality, suppose it is $Y$. Now consider the following cases. 

\begin{enumerate}
\item If none of the goods in $Y$ is allocated to agent $1$, then $\mathcal{A}$ must allocate $X$ to agent $2$ and $Y$ to agent $3$. 
\item Suppose that exactly one of the goods in $Y$, say good $5$, is allocated to agent $1$. In this case, observe that $4$ must be allocated. This is because if $4$ is unallocated, this means that both agents $2$ and $3$ receive at least one good from $X$ and $6$ must be allocated to agent $1$, a contradiction. 
Thus, the only possible maximal allocations are $(\{5,7\},\{1,2,3\},\{4,6\})$ and $(\{5,7\},\{6\},\{4\})$.
A similar argument holds when $6$ is allocated to agent $1$ due to the symmetry of the graph. 

\item Suppose that two goods, $5$ and $6$, in $Y$ are allocated to agent $1$. If $4$ is allocated, then the only possible maximal allocation is $(\{5,6,7\},\{1,2,3\},\{4\})$. 
If $4$ is unallocated, this means that both agents $2$ and $3$ receive at least one good from $X$ and all goods in $X$ must be allocated to either of such agents. Thus, the only possible maximal allocations are $(\{5,6,7\},\{1\},\{2,3\})$, $(\{5,6,7\},\{1,2\},\{3\})$, together with $(\{5,6,7\},\{1,3\},\{2\})$. 
\end{enumerate}

    For this reason, there are only $6$ maximal allocations to consider (refer to Figure \ref{fig:counterexample_n3}). All of them are not EF1, because:
	\begin{itemize}
		\item Allocations \#1, \#4, \#5, \#6: one agent takes one good but there is another agent who takes three goods
		\item Allocation \#2: $v(\{5, 7\}) = 3$ but $v(\{1, 2\}) = v(\{1, 3\}) = v(\{2, 3\}) = 4$.
		\item Allocation \#3: $v(4) = 1$ but $v(5) = v(7) = 2$.~\qedhere
	\end{itemize}
\end{proof}

\begin{figure}[htbp]
	\centering
	\includegraphics[width=0.5\textwidth]{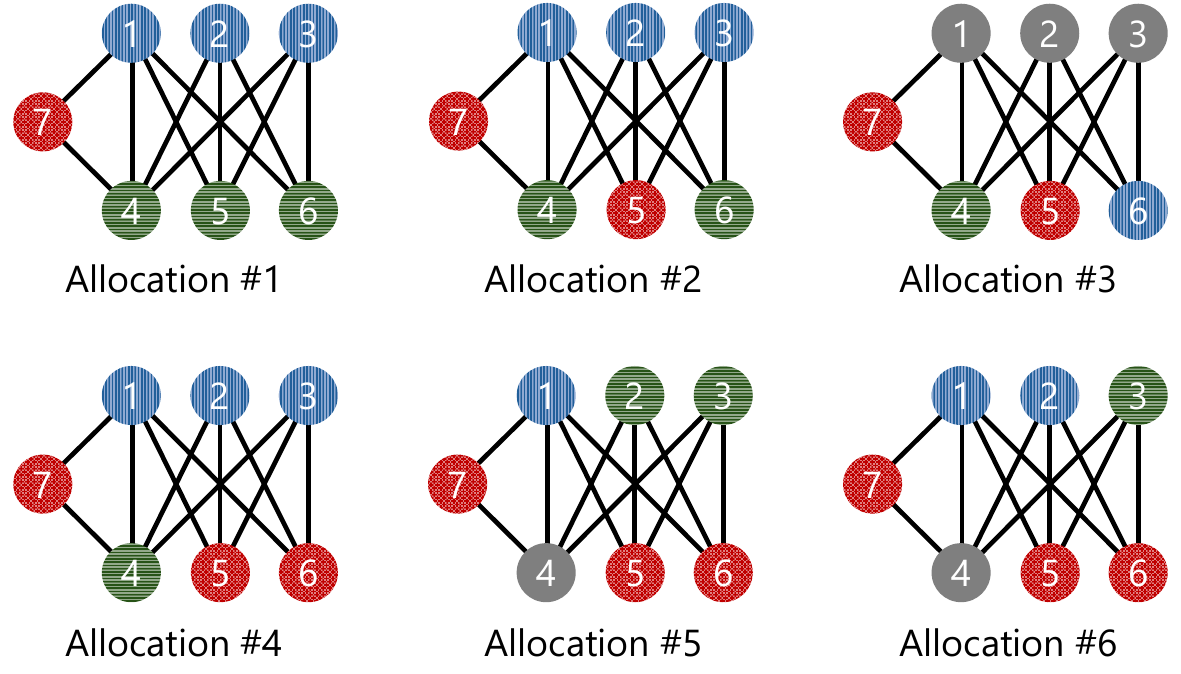}
	\caption{$6$ maximal allocations to consider in the instance of Theorem \ref{thm:counterexample_n3}. The vertices in red, blue, green, and gray are goods taken by agent $1, 2, 3$, and no one, respectively.}
	\label{fig:counterexample_n3}
\end{figure}

The instance constructed in the proof of Theorem~\ref{thm:counterexample_n3} uses an identical monotone valuation. It remains an open question whether a counterexample exists for three agents with (even identical) additive valuations.%\footnote{We did not find such an instance for $m \leq 8$ via an intensive computer experiment}.

For every number $n \geq 4$ of agents, Hummel and Hetland~\cite{HummelH22} presented an instance with identical additive valuations and a complete bipartite graph $K_{n-1,n-1}$ for which a complete EF1 allocation does not exist.
In such cases, a complete allocation does exist, as the maximum degree of the graph is less than the number of agents. Thus, this counterexample implies that achieving both EF1 and maximality is impossible even when four agents have identical additive valuations. Here, we provide a smaller example using $K_{3,n-1}$, which turns out to be smallest since for $m\leq n+1$, there always exists a maximal EF1 allocation. 

%Hummel and Hetland \cite{HummelH22} found an instance that a maximal EF1 allocation does not exist; $n = 4, m = 6, G = K_{3, 3}$ and the valuation is additive and identical which values $2$ each for one part of the bipartite graph and $3$ each for the another part. We generalize this example to all $n \geq 4$.\footnote{In fact, Hummel and Hetland \cite{HummelH22} mentioned an instance that $G = K_{n-1, n-1}$, but the example in \textbf{Theorem \ref{thm:counterexample_n4}} is smaller.}

\begin{restatable}{proposition}{ExampleFourAgents}
 For every number $n \geq 4$ of agents, there is an instance with identical additive valuations, $m=n+2$, and $G=K_{3,n-1}$ where no maximal EF1 allocation exists.
    \label{thm:counterexample_n4}
\end{restatable}
\begin{proof}
    Consider the following instance with $m = n+2$ goods and a complete bipartite graph $K_{3, n-1}$ with the left part containing three goods and the right part containing $n-1$ goods. Let $a_1, a_2, a_3$ denote the three goods on the left side and $b_1,b_2,\ldots,b_{n-1}$. Each agent has an identical additive valuation where
        \begin{equation*}
            v(g) = \begin{cases}
                2 & (g \in \{a_1, a_2, a_3 \}) \\
                3 & (\text{otherwise}).
            \end{cases}
        \end{equation*}
    Consider an arbitrary maximal allocation $\mathcal{A}$. 
    Since the maximum degree of the graph is $n-1$, all goods must be allocated in this allocation. Due to pigeonhole's principle, there exists an agent X that cannot receive any of goods on the right.
    \begin{itemize}
        \item Case 1: If there is only one such agent, goods $b_1,b_2,\ldots,b_{n-1}$ are allocated to $n-1$ other agents, which means that each of such agents takes one good of value $3$. The remaining goods $a_1, a_2, a_3$ must be taken by agent X. Thus, agent X is envied by each of the other agents even after removing one of the goods $a_1, a_2, a_3$ from X's bundle.
        \item Case 2: If there are two or more such agents, there exists an agent Y who takes two or more goods among goods $b_1,b_2,\ldots,b_{n-1}$; further, there exists an agent Z who does not receive any of the goods on the right and takes at most one good from the left side. Then, agent Y is envied by agent Z after removing one of the goods from Y's bundle.
    \end{itemize}
    Therefore, in either case, the allocation $\mathcal{A}$ is not EF1.
\end{proof}

%The instance given \textbf{Theorem \ref{thm:counterexample_n4}} can be proved to be smallest.

\begin{restatable}{proposition}{nPlusOne}
For $n$ agents with monotone valuations and $m$ goods with $m \leq n+1$, there exists a maximal EF1 allocation.
    \label{prop:nPlusOne}
\end{restatable}
\begin{proof}
Suppose $m \leq n+1$. Consider the following algorithm. 
%When $m \leq n$, a maximal EF1 allocation exists because we can allocate good $g_1, \dots, g_m$ to agent $a_1, \dots, a_m$, respectively. 
%When $m = n+1$, we can obtain a maximal EF1 allocation using the following algorithm:
	\begin{enumerate}
		\item For each agent $i= 1, 2, \dots, n$ in this order, run one cycle of the round-robin algorithm. Namely, agent $i$ chooses the most preferred good $g^*$ among the remaining goods (namely, $g^* \in \argmax_{g \in R}v_i(g)$ where $R$ is the set of remaining goods). If there is no remaining good, the agent does not receive any good.  
		\item After that, at most one good $g$ remains. If $g$ can be feasibly assigned to some agent $i$ (namely, $g$ is not adjacent to the good received by $i$ in the first phase), assign it to any of such agents.
	\end{enumerate}
	It is obvious that the resulting allocation is maximal. If $m \leq n$, the resulting allocation is clearly EF1 since each agent receives at most one good. If $m = n+1$, consider an agent $i$ who gets two goods. After removing the first good she has chosen, each of the other agents does not envy $i$ since the last good $g$ was not selected in the first phase.  
\end{proof}

% \pasin{Maybe we should (i) add an illustration of this graph, and (ii) produce a table of all these 18 allocations and, for each allocation, either who is envying whom. Or if there is a more human-verifiable proof, that would be nice too.}

\subsection{NP-Hardness}

Now, we will show that it is NP-hard to decide whether a maximal EF1 allocation exists, as stated below.

\begin{theorem}
Given the graph and valuations, determining whether a maximal EF1 allocation exists is NP-hard for:
\begin{enumerate}
\item any fixed $n \geq 4$, even when restricted to identical and additive valuation, and,
\item $n = 3$, even when restricted to identical and monotone valuation.
\end{enumerate}
\label{thm:nphard}
\end{theorem}

In fact, our proof can transform any negative example into an NP-hardness result, as stated more precisely below.

%We prove the above NP-hardness by using examples in the last section as gadgets for a reduction from the independent set (IS) problem. In the IS problem, we are given a graph $H = (V_H, G_H)$ and a positive integer $t$, and the goal is to decide whether $H$ contains an IS of size $t$.

\begin{lemma} \label{lem:main-red}
Suppose that there exists an instance $\tI = ([n], \tM, \tv, \tG = (\tM, \tE))$ (with identical valuation) where the number $n$ of agents and the number $|\tM|$ of items are both constants, such that no maximal EF1 allocation exists. Then, it is NP-hard to decide whether a maximal EF1 allocation exists for $n$ agents with identical valuations. 
Furthermore, if $\tv$ is additive, then this NP-hardness applies even when the valuations are additive. 
%Then, there exists a polynomial-time algorithm that takes in any instance $(H = (V_H, G_H), t)$ of IS and outputs an instance $I = ([n], M, v, G)$  (with identical valuation) such that the following holds:
%\begin{itemize}
%\item (Completeness) If $H$ contains an independent set of size at least $c \cdot t$, $I$ contains a maximal EF1 allocation.
%\item (Soundness) If $H$ does not contains an independent set of size $t$, $I$ does not contain a maximal EF1 allocation.
%\end{itemize}
%Furthermore, if the valuations of $\tI$ is additive, then the valuations of $I$ is also additive.]
\end{lemma}

\Cref{thm:nphard} is an immediate corollary of  \Cref{lem:main-red} where $\tI$ is the instance from \Cref{thm:counterexample_n4} or \Cref{thm:counterexample_n3}. Note that we state \Cref{lem:main-red} in this generic form so that, if subsequent work finds such an instance $\tI$ for additive valuation for $n = 3$, then the NP-hardness would follow as a corollary.

\paragraph{Reduction.}
We will reduce from the Independent Set (IS) problem, which is one of Karp's classic NP-hard problems~\cite{Karp72}. In the IS problem, we are given a graph $H = (V_H, E_H)$ and a positive integer $t$, and the goal is to decide whether $H$ contains an IS of size $t$.

At a high-level, our reduction starts from $\tI$ and adds to it $n$ copies of the graph $H$, where each good has the same value $\lambda$. Roughly speaking, we wish the $i$-th copy of $H$ (denoted by $X_i$ in the proof below) to give ``extra goods'' to the $i$-th agent, in case that agent envies some other agent by more than one good. The crux of the reduction is that such an agent can ``catch up'' (and thus satisfy EF1) iff there is a sufficiently large independent set in $H$.  This is not yet a complete reduction since, $H$ may not have a \emph{maximal} independent set of a certain prescribed size. To alleviate this, we introduce ``dummy goods'' with zero value (denoted by $Y_i$ below) to ensure that we can pick any desired number of goods from each copy of $H$. Finally, some additional edges are also added to ensure that each agent selects goods from a single copy of $H$.

For the proof below, we will use the following notation: for any valuation $v$ and set $S$ of goods, let $v^{-1}(S) := \min_{j \in S} v(S \setminus \{j\})$ denote the value of $S$ after its most valuable good is removed. We use the convention $v^{-1}(\emptyset) = 0$.

\begin{proof}[Proof of \Cref{lem:main-red}]
Recall $\tI$ from the lemma statement.
Let\footnote{$\gamma$ can be computed in $O(1)$ time by bruteforce.}%enumerating all maximal allocations in $\tI$.} 
$\gamma := \min_{\tcA} \max_{i, i' \in [n]} \left(\tv^{-1}(\tA_i) - \tv(\tA_{i'})\right)$ where the outer minimum is over all maximal allocation $\tcA$ of $\tI$. By the assumption on $\tI$, we have $\gamma > 0$. Let $\lambda := \gamma / t$.

Let $(H = (V_H, E_H), t)$ denote the IS instance. Our reduction constructs the instance $I = ([n], M, v, G)$ as follows:
\begin{itemize}
\item {\bf Goods}: $M = \tM \cup X_1 \cdots \cup X_n \cup Y_1 \cup \cdots \cup Y_n$ where $X_i = \{x_{i, w} \mid w \in V_H\}$ and $Y_i = \{y_{i, w} \mid w \in V_H\}$ are sets (each of size $|V_H|$) of additional goods.
\item {\bf Graph}: $G$ contains the following edges:
	\begin{enumerate}[(i)]
	\item All edges in $\tG$,
	\item $(x_{i, u}, x_{i, w})$ for all $i \in [n]$ and $(u, w) \in V_H$,
        \item $(x_{i, w}, y_{i, w})$ for all $i \in [n]$ and $w \in V_H$,
        \item all pairs of vertices in $(X_i \cup Y_i) \times (X_{i'} \cup Y_{i'})$ for all distinct $i, i' \in [n]$.
	\end{enumerate}
	\item {\bf Valuation}: For all $S \subseteq M$, let $v(S) = \tv(S \cap \tM) + \lambda |S \cap X|$ where $X := X_1 \cup \cdots \cup X_n$. That is, the valuations on original goods remain the same, the valuations of each good in $X$ is $\lambda$, and the goods in $Y_1 \cup \cdots \cup Y_n$ have valuations zero.
\end{itemize}

See Figure \ref{fig:nphard_construction} for an illustration of the instance $I$.

%We refer to goods in $X_i$ as ``extra goods of group $i$'' and $Y_i$ as ``dummy goods of group $i$''.

%We call $x_{1, 1}, \dots, x_{n, k}$, which are valued $\lambda$ each, as "extra goods", and call $x_{i, 1}, \dots, x_{i, k}$ as "group $i$".

It is clear that the reduction runs in polynomial time, and that, if $\tv$ is additive, then $v$ is also additive.

\paragraph{(YES)} Suppose that $H$ contains an IS of size $t$.
Let $\tcA^*$ be a maximal allocation of $\tI$ such that $\max_{i, i' \in [n]} \left(\tv^{-1}(\tA^*_i) -  \tv(\tA^*_{i'})\right) = \gamma$. We may assume w.l.o.g. that $v^{-1}(\tA^*_1) \geq v^{-1}(\tA^*_2), \dots, v^{-1}(\tA^*_n)$. For each $i \in [n]$, we construct $A^*_i$ as follows:
\begin{enumerate}
\item  Let $c_i := \lceil \max\{0, \tv^{-1}(\tA^*_1) - \tv(\tA_i^*)\} / \lambda \rceil \leq t$.
\item Let $S_i$ be any (non-necessarily maximal) IS of size $c_i$ in $H$, which exists since $H$ contains an IS of size $t$.
\item Let $A^*_i = \tA^*_i \cup \{x_{i, v}\}_{v \in S_i} \cup \{y_{i, v}\}_{v \in (V_H \setminus S_i)}$
\end{enumerate}
Observe that each good in $X_i \cup Y_i$ can only belong to $A^*_i$, and it is obvious that there is no edge between goods in $A^*_i$. Thus, $\cA^* = (A^*_1, \dots, A^*_n)$ is a valid allocation. To see that this is maximal, note that the goods from $Y_i$ (resp., $X_i$) together with type-(iii) edges ensure that no other goods in $X_i$ (resp., $Y_i$) can be added to $A^*_i$. Since at least one good from $X_i \cup Y_i$ is picked, type-(iv) edges ensure that no goods in $X_{i'} \cup Y_{i'}$ for $i' \ne i$ can be added to $A^*_i$.

Finally, we argue that $\cA^*$ is EF1. 
To bound $v^{-1}(A^*_i)$, note that $v(A^*_{i}) = \tv(\tA^*_{i}) + c_i \lambda$. Consider two cases based on $c_i$.
\begin{itemize}
\item If $c_i = 0$, we have $v^{-1}(A^*_i) = \tv^{-1}(\tA^{*}_i) \leq \tv^{-1}(\tA^*_1)$.
\item If $c_i > 0$, by definition of $c_i$, we have $v(A^*_i) < \tv^{-1}(\tA^*_1) + \lambda$. Thus, $v^{-1}(A^*_i) \leq v(A^*_i) - \lambda < \tv^{-1}(\tA^*_1)$.
\end{itemize}
Thus, in both cases, we have $v^{-1}(A^*_i) \leq \tv^{-1}(\tA^*_1)$. 

On the other hand, for any $i' \in [n]$, the definition of $c_{i'}$ immediately implies $v(A^*_{i'}) \geq \tv^{-1}(\tA^*_1)$.

By the two previous paragraphs, $\cA^*$ is EF1.

\paragraph{(NO)} Suppose that $H$ does not contain an IS of size $t$. Consider any maximal allocation $\cA = (A_1, \dots, A_n)$ of $I$. Notice that the allocation $\tcA = (A_1 \cap \tM, \dots, A_n \cap \tM)$ is maximal w.r.t. $\tI$. Thus, there exist $i, i' \in [n]$ such that $\tv^{-1}(\tA_i) - \tv(\tA_{i'}) \geq \gamma$. Due to type-(iv) edges, at most one of $A_{i'} \cap X_1, \dots, A_{i'} \cap X_n$ can be non-empty. Furthermore, type-(ii) edges imply that the non-empty set must correspond to an independent set in $H$. From our assumption, this implies $|A_{i'} \cap X| < t$. As a result,
\begin{align*}
v^{-1}(A_i) \geq \tv^{-1}(A_i)
&\geq \gamma + \tv(\tA_{i'}) \\
&= \gamma + v(A_{i'}) - \lambda |A_{i'} \cap X| \\
&> \gamma + v(A_{i'}) - \lambda \cdot t  \qquad \overset{(\star)}{\geq} v(A_{i'}),
\end{align*}
where $(\star)$ is due to our choice of $\lambda$. Thus, $\tA$ is not EF1.
\end{proof}

\begin{figure}[htbp]
	\centering
	\includegraphics[width=0.5\textwidth]{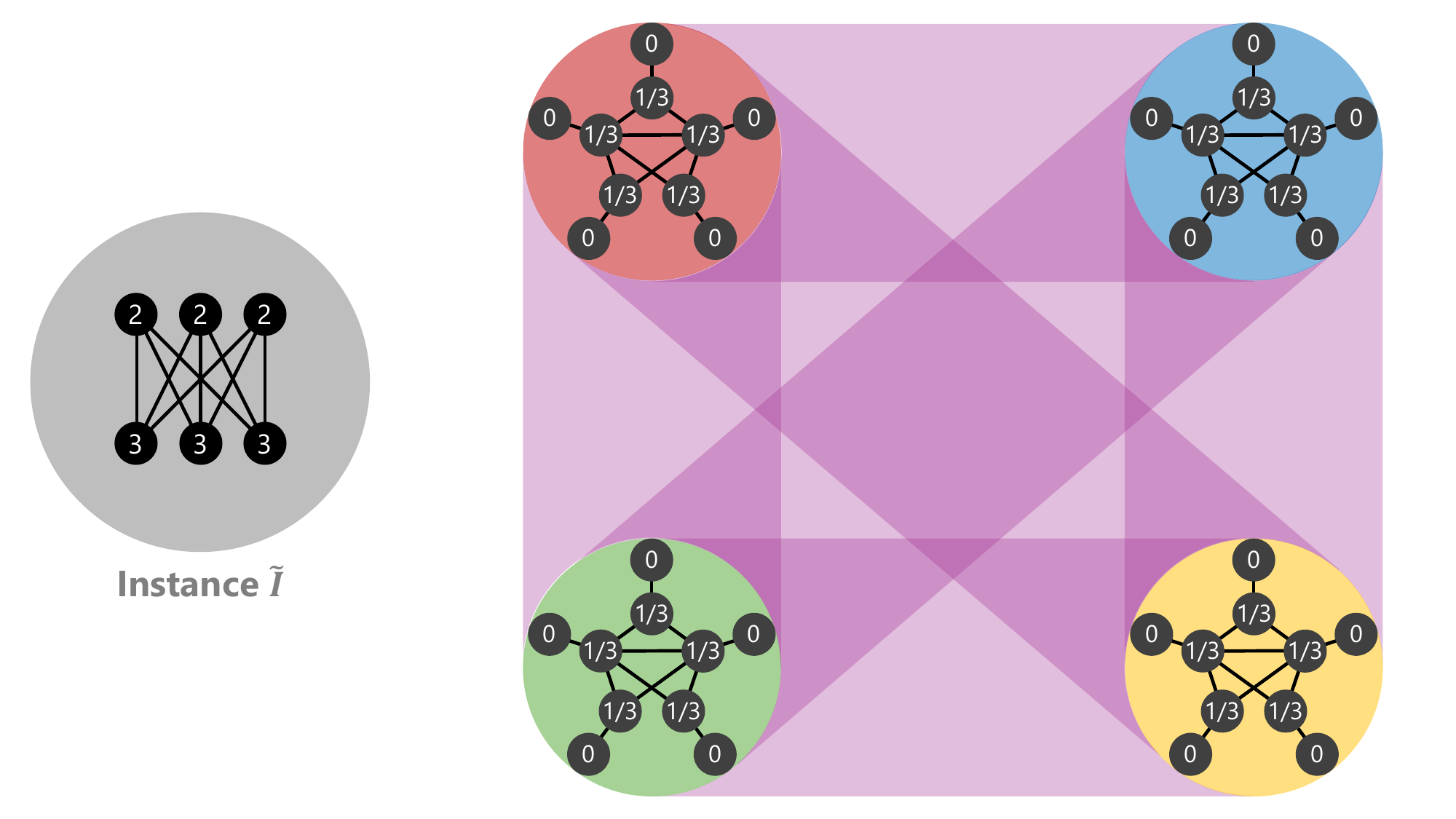}
	\caption{Instance $I$ created with $n = 4$ instance given by Proposition \ref{thm:counterexample_n4} for $\tI$, and $5$-vertex $7$-edge graph for $H$, setting $t = 3$ (note that $\gamma = 1, \lambda = \frac{1}{3}$). The bands in purple represents type-(iv) edges.}.
	\label{fig:nphard_construction}
\end{figure}

\section{Chore Allocation}

In this section, we consider the chore version of our problem, where each agent $i$ has a \emph{monotone non-increasing} valuation function $v_i$, namely, $v_i(S) \geq v_i(T)$ for every $S \subseteq T \subseteq M$. %For example, the task of assigning jobs that has designated start and finish time (jobs must be assigned so that the working times do not overlap), fits to this problem. In chore version, the valuation is negative or zero, and the definition of EF1 is different. 
An allocation $\mathcal{A} = (A_1, \dots, A_n)$ is \emph{envy-free up to one chore} (EF1 for chores) if for every pair of agents $i,j \in N$, either $A_i = \emptyset$, or $v_i(A_i \setminus \{c\}) \geq v_i(A_j)$ some $c \in A_i$~\cite{aziz2022fair,Bhaskar2021}. 
For identical valuations, it is easy to see that the existence of a maximal EF1 allocation is equivalent for goods and chores.
Specifically, an allocation $\mathcal{A}$ is envy-free up to one chore under valuation $v$ if and only if $\mathcal{A}$ is envy-free up to one good under valuation $-v$. Thus, a maximal allocation that is envy-free up to one chore exists for $v$ if and only if a maximal allocation that is envy-free up to one good exists for $-v$. 
%\begin{itemize}
%	\item The condition for $v'$ to be EF1 is that, for all $a_i, a_j \in N$ with $i \neq j$, either of the following holds: (1) $-v(A_i \setminus \{g\}) \geq -v(A_j)$ for some $g \in A_j$, (2) $A_i \neq 0$.
%	\item The condition for $v$ to be EF1 is that, for all $a_j, a_i \in N$ with $j \neq i$, either of the following holds: (1) $v(A_j) \geq v(A_i \setminus \{g\})$ for some $g \in A_i$, (2) $A_i \neq 0$.
%\end{itemize}
%These two conditions are apparently equivalent.

Consider the case of two agents with monotone non-increasing valuations. Theorem~\ref{thm:cutchoose} holds in this case, so we can assume without loss of generality that the agents have identical valuations $v$. By combining the discussion above with the fact that $-v$ is monotone non-decreasing, Theorem \ref{thm:twoagents_main} guarantees the existence of a maximal EF1 allocation. Moreover, the algorithms presented in \Cref{thm:generic-alg,thm:twoagents_bipartite,thm:twoagents_interval} remain applicable in this setting.

For three or more agents with monotone non-increasing valuations, there exist instances where a maximal allocation does not exist. In fact, the instances given in Theorem~\ref{thm:counterexample_n4} and Proposition~\ref{thm:counterexample_n3} use an identical valuation $v$, and the corresponding instance obtained by replacing $v$ with $-v$ yields no maximal allocation that is EF1 for chores. Also, determining whether a maximal EF1 allocation exists or not is NP-hard also for monotone non-increasing valuations. This follows from Theorem \ref{thm:nphard} since the constructed instance uses an identical monotone non-decreasing valuation. 
%For uniform valuation ($v(S) = -\left|S\right|$ for all $S \subseteq M$), there exists a maximal EF1 allocation when $G$ is an interval graph or a tree, due to Theorem \ref{thm:uniform_interval} and \ref{thm:uni-tree}, and that $v$ has identical valuation.

\section{Conclusion}
While we give a nearly complete picture of the existence for maximal EF1 allocations for monotone valuations, there are still a few open interesting questions left. First, for $n = 2$, our algorithm (\Cref{thm:generic-alg}) runs in pseudo-polynomial time for general monotone valuations. Is there a polynomial-time (in $ m, T(m)$) algorithm for this task?

It might also be worthwhile considering special cases, which can result from restricting either the valuations or the graphs. We list a few intriguing cases below. 
\begin{itemize}
\item Additive valuations: The existence of maximal EF1 allocation remains open only for the case $n = 3$ since our lower bound in \Cref{thm:counterexample_n3} requires non-additive valuations, but the lower bound for $n \geq 4$ (\Cref{thm:counterexample_n4})  holds for additive (and identical) valuations.
\item Uniform valuations: All $n \geq 3$ remains open in this case as we are not aware of any lower bound that holds for uniform valuations. We also note that the well-known Hajnal-Szemerédi theorem (\cite{HS70}) is equivalent to stating that a maximal EF1 allocation exists when the graph has maximum degree at most $n - 1$. Thus, a positive answer to the question for arbitrary graphs will significantly generalize the theorem. We make progress on this question by proving the existence on trees (see Appendix~\ref{sec:uniform}).
\item Restricted Graph Classes: It might also be interesting to study special graph classes, such as bounded degree graphs (as suggested by the above Hajnal-Szemerédi theorem) or trees. Some other interesting graph classes include interval graphs--which was studied by Kumar et al.~\cite{KEG+24}--and bipartite graphs. 
\end{itemize}

\section*{Acknowledgements}
This work was partially supported by JST FOREST Grant Numbers JPMJPR20C1. We thank the anonymous IJCAI 2025 for their valuable comments.

%\clearpage
\bibliographystyle{abbrv}%alpha}%{plainnat} 
%\bibliography{abb,references}

\begin{thebibliography}{99}

\bibitem{AAB+23fair}
G.~Amanatidis, H.~Aziz, G.~Birmpas, A.~Filos-Ratsikas, B.~Li, H.~Moulin, A.~A. Voudouris, and X.~Wu.
\newblock {Fair division of indivisible goods: recent progress and open questions}.
\newblock {\em Artificial Intelligence}, page 103965, 2023.

\bibitem{aziz2022fair}
H.~Aziz, I.~Caragiannis, A.~Igarashi, and T.~Walsh.
\newblock Fair allocation of indivisible goods and chores.
\newblock {\em Autonomous Agents and Multi-Agent Systems}, 36(1):1--21, 2022.

\bibitem{Bhaskar2021}
U.~Bhaskar, A.~R. Sricharan, and R.~Vaish.
\newblock {On Approximate Envy-Freeness for Indivisible Chores and Mixed Resources}.
\newblock In {\em Proceedings of Approximation, Randomization, and Combinatorial Optimization. Algorithms and Techniques (APPROX/RANDOM 2021)}, pages 1:1--1:23, 2021.

\bibitem{BiloCaFl22}
V.~Bil\`{o}, I.~Caragiannis, M.~Flammini, A.~Igarashi, G.~Monaco, D.~Peters, C.~Vinci, and W.~S. Zwicker.
\newblock Almost envy-free allocations with connected bundles.
\newblock {\em Games and Economic Behavior}, 131:197--221, 2022.

\bibitem{BiswasFORC2023}
A.~Biswas, Y.~Ke, S.~Khuller, and Q.~C. Liu.
\newblock {An algorithmic approach to address course enrollment challenges}.
\newblock In {\em Proceedings of 4th Symposium on Foundations of Responsible Computing (FORC 2023)}, pages 8:1--8:23, 2023.

\bibitem{BouveretCeEl17}
S.~Bouveret, K.~Cechl\'{a}rov\'{a}, E.~Elkind, A.~Igarashi, and D.~Peters.
\newblock Fair division of a graph.
\newblock In {\em Proceedings of the 26th International Joint Conference on Artificial Intelligence (IJCAI)}, pages 135--141, 2017.

\bibitem{BramsFi00}
S.~J. Brams and P.~C. Fishburn.
\newblock Fair division of indivisible items between two people with identical preferences: Envy-freeness, {P}areto-optimality, and equity.
\newblock {\em Social Choice and Welfare}, 17(2):247--267, 2000.

\bibitem{brams2014two}
S.~J. Brams, M.~Kilgour, and C.~Klamler.
\newblock Two-person fair division of indivisible items: An efficient, envy-free algorithm.
\newblock {\em Notices of the AMS}, 61(2):130--141, 2014.

\bibitem{BramsTa96}
S.~J. Brams and A.~D. Taylor.
\newblock {\em Fair division: from cake-cutting to dispute resolution}.
\newblock Cambridge University Press, 1996.

\bibitem{Budish11}
E.~Budish.
\newblock The combinatorial assignment problem: {A}pproximate competitive equilibrium from equal incomes.
\newblock {\em Journal of Political Economy}, 119(6):1061--1103, 2011.

\bibitem{BudishCaKe17}
E.~Budish, G.~P. Cachon, J.~B. Kessler, and A.~Othman.
\newblock Course {M}atch: a large-scale implementation of approximate competitive equilibrium from equal incomes for combinatorial allocation.
\newblock {\em Operations Research}, 65(2):314--336, 2017.

\bibitem{ChiarelliKMPPS20}
N.~Chiarelli, M.~Krnc, M.~Milanic, U.~Pferschy, N.~Pivac, and J.~Schauer.
\newblock Fair packing of independent sets.
\newblock {\em Algorithmica}, 85(5):1459--1489, 2023.

\bibitem{F67resource}
D.~Foley.
\newblock {Resource allocation and the public sector}.
\newblock {\em Yale Economic Essays}, pages 45--98, 1967.

\bibitem{GS58puzzle}
G.~Gamow and M.~Stern.
\newblock {Puzzle-Math}.
\newblock {\em Viking Press}, 1958.

\bibitem{GoldmanPr14}
J.~Goldman and A.~D. Procaccia.
\newblock Spliddit: Unleashing fair division algorithms.
\newblock {\em {ACM} {SIG}ecom Exchanges}, 13(2):41--46, 2014.

\bibitem{HS70}
A.~Hajnal and E.~Szemer\'edi.
\newblock Proof of a conjecture of {P}. {E}rd{\H{o}}s.
\newblock {\em Colloq Math Soc János Bolyai}, 4:601--623, 01 1970.

\bibitem{Hsu93}
W.-L. Hsu.
\newblock A simple test for interval graphs.
\newblock In {\em Graph-Theoretic Concepts in Computer Science: 18th International Workshop, WG'92 Wiesbaden-Naurod, Germany, June 18–20, 1992 Proceedings 18}, pages 11--16, 1993.

\bibitem{HummelH22}
H.~Hummel and M.~L. Hetland.
\newblock Fair allocation of conflicting items.
\newblock {\em Autonomous Agents and Multi-Agent Systems}, 36(1):8, 2022.

\bibitem{IgarashiYo23}
A.~Igarashi and T.~Yokoyama.
\newblock Kajibuntan: a house chore division app.
\newblock In {\em Proceedings of the 37th AAAI Conference on Artificial Intelligence (AAAI)}, pages 16449--16451, 2023.

\bibitem{Karp72}
R.~M. Karp.
\newblock Reducibility among combinatorial problems.
\newblock In R.~E. Miller and J.~W. Thatcher, editors, {\em Proceedings of a symposium on the Complexity of Computer Computations, held March 20-22, 1972, at the {IBM} Thomas J. Watson Research Center, Yorktown Heights, New York, {USA}}, The {IBM} Research Symposia Series, pages 85--103, 1972.

\bibitem{KEG+24}
Y.~Kumar, S.~Equbal, R.~Gurjar, S.~Nath, and R.~Vaish.
\newblock Fair scheduling of indivisible chores.
\newblock In {\em Proceedings of the 23rd International Conference on Autonomous Agents and Multi-Agent Systems (AAMAS)}, pages 2345--2347, 2024.
\newblock Extended version available at \url{https://arxiv.org/abs/2402.04353}.

\bibitem{LiFairScheduling2021}
B.~Li, M.~Li, and R.~Zhang.
\newblock Fair scheduling for time-dependent resources.
\newblock In {\em Proceedings of the 34th Annual Conference on Neural Information Processing Systems (NeurIPS)}, 2021.

\bibitem{LiptonMaMo04}
R.~J. Lipton, E.~Markakis, E.~Mossel, and A.~Saberi.
\newblock On approximately fair allocations of indivisible goods.
\newblock In {\em Proceedings of the 5th ACM Conference on Electronic Commerce (EC)}, pages 125--131, 2004.

\bibitem{Moon1965OnCI}
J.~W. Moon and L.~Moser.
\newblock On cliques in graphs.
\newblock {\em Israel Journal of Mathematics}, 3:23--28, 1965.

\bibitem{Steinhaus48}
H.~Steinhaus.
\newblock Sur la division pragmatique.
\newblock {\em Econometrica}, 17:315--319, 1949.

\bibitem{Suksompong21}
W.~Suksompong.
\newblock Constraints in fair division.
\newblock {\em ACM SIGecom Exchanges}, 19(2):46--61, 2021.

\end{thebibliography}

\appendix
\newpage

\begin{center}
    \Large{Appendix}    
\end{center}

\section{Two agents and Special Graph Classes}\label{appx:twoagents}

%\paragraph{Example of a graph.}
%\smallskip
%~~
%\begin{figure}[htbp]%
%	\centering
%	\includegraphics[width=0.5\textwidth]{figure-33.pdf}
%	\caption{}
%	\label{fig:notreachable}
%\end{figure}

\paragraph{Bipartite Graphs.}
For two agents and special graph classes like bipartite graphs, we can find a maximal EF1 and allocation in polynomial time.

\begin{theorem}
When there are $n = 2$ agents and $G$ is bipartite, a maximal EF1 allocation can be found in $O\left(\left|E\right| + m \cdot T(m)\right)$ time. %, where $T(m)$ is the time to compute valuation $v(S)$ for a given $S \subseteq A$. Note that this is a polynomial-time algorithm when $T(m)$ is polynomial time.
	\label{thm:twoagents_bipartite}
\end{theorem}

\begin{proof}
Let $M_1, M_2$ the two parts of bipartite graphs that $v(M_1) \geq v(M_2)$; we may assume w.l.o.g. that $M_1$ contains all vertices of degree zero. The algorithm is simply to run \Cref{alg:ef1-chain} with $S = M_1$. The running time claim is obvious. Meanwhile, since $X_1, X_2 \subseteq M_2$, we have that $v(X_1), v(X_2) \leq v(M_2) \leq v(M_1) = v(S)$. Thus, \Cref{lem:alg-chain} guarantees that the algorithm finds an EF1 allocation.
%If we set $S = M_1$ in the proof of \textbf{Theorem \ref{thm:twoagents_main}}, then the resulting $X_1, X_2$ is $M_2$, satisfying the conditions $v(S) \geq v(X_1)$ and $v(S) \geq v(X_2)$. When $S, X_1, X_2$ are decided, the EF1 and maximal allocation can be computed in $O(m \cdot T(m))$ time (as in the proof of \textbf{Theorem \ref{thm:twoagents_main}}), the total time to compute EF1 and maximal allocation can be bounded by $O(\left|E\right| + m \cdot T(m))$. Note that it takes $O(m + \left|E\right|)$ time to compute partition $M_1, M_2$ from graph $G$, by using DFS or BFS.
\end{proof}

\paragraph{Interval Graphs.}
For interval graphs, Kumar et al.~\cite{KEG+24} proved that a maximal EF1 allocation exists among two agents and can be computed in polynomial time. They extended the idea presented in Figure \ref{fig:chain_pathgraph} for path graphs to interval graphs. However, unlike paths, there is an issue that the allocation after the slight change may not be maximal. Kumar et al. addressed this issue by dividing the structure of intervals into five cases, showing that there is a way to amend the allocation to satisfy maximality for each case. However, their proof is technically involved. Here, we provide a simpler proof by utilizing Theorem~\ref{thm:twoagents_main}.

\begin{restatable}{theorem}{IntervalTwo}
	When there are $n = 2$ agents with monotone valuations and $G$ is an interval graph, a maximal EF1 allocation can be found in $O(m \log m + m \cdot T(m))$ time.
	\label{thm:twoagents_interval}
\end{restatable}

%\IntervalTwo*

We consider (generalized) interval scheduling problem as a base.

\begin{description}
	\item[\textsc{IntervalScheduling}.] Given $n$ intervals $[l_1, r_1), \dots,$ $[l_n, r_n)$ and an integer $c \geq 1$, find a subset of intervals with maximum size such that no point is covered by more than $c$ intervals. For simplicity, we assume that $l_1, r_1, \dots, l_n, r_n$ are all distinct.
\end{description}

This can be solved by greedy algorithm: First, we sort intervals so that $r_1 < \dots < r_n$. Then, for $i = 1, \dots, n$ in this order, if $[l_i, r_i)$ can be added while satisfying the constraint that each point is  covered by at most $c$ intervals, add this to the current solution.

Let's understand the structure of greedy solution. Let $[l'_1, r'_1), \dots, [l'_{k'}, r'_{k'}) \ (r'_1 < \dots < r'_{k'})$ be any valid solution. For convenience, $l'_i, r'_i$ for $i \geq k'+1$ are set to $+\infty$. Also, let $[l^*_1, r^*_1), \dots, [l^*_k, r^*_k) \ (r^*_1 < \dots < r^*_k)$ be the greedy solution.

\begin{lemma}
	$r^*_i \leq r'_i$ holds for any $i \in \{1, \dots, k\}$.
	\label{lem:interval_lemma1}
\end{lemma}

\begin{proof}
	We prove by induction of $i$.
	\begin{itemize}
		\item \textbf{Base Case $(i = 1)$}: $r^*_1 \leq r'_1$ holds, because the greedy algorithm first chooses the interval $[l_i, r_i)$ with the smallest $r_i$.
		\item \textbf{Induction $(i \geq 2)$}: We assume that $r^*_1 \leq r'_1, \dots, r^*_{i-1} \leq r'_{i-1}$ holds. Suppose that $r^*_i > r'_i$. Because each point is covered by at most $c$ intervals, $r'_{i-c} \leq l'_i$ must hold. The assumption $r^*_{i-c} \leq r'_{i-c}$ implies that $r^*_{i-c} \leq l'_i$, so $[l^*_1, r^*_1), \dots, [l^*_{i-1}, r^*_{i-1}), [l'_i, r'_i)$ is a valid solution. Because of the construction of the greedy algorithm, $[l'_i, r'_i)$ should have been added before $[l^*_i, r^*_i)$, a contradiction. Therefore, $r^*_i \leq r'_i$.~\qedhere
	\end{itemize}
\end{proof}
The above result also implies that the greedy solution is optimal, i.e. $k \geq k'$.

\begin{lemma}
	Consider the interval scheduling problem with $c = 1$. Suppose that $[l'_1, r'_1), \dots, [l'_k, r'_k)$ is an optimal solution. Then, for any $i \in \{0, \dots, k\}$, $[l^*_1, r^*_1), \dots, [l^*_{i-1}, r^*_{i-1}), [l'_i, r'_i), \dots, [l'_k, r'_k)$ is an optimal solution.
	\label{lem:interval_lemma2}
\end{lemma}

\begin{proof}
	The inequality that $r^*_{i-1} \leq r'_{i-1}$ holds due to \textbf{Lemma \ref{lem:interval_lemma1}}, and $r'_{i-1} < l'_i$  holds because $[l'_{i-1}, r'_{i-1})$ and $[l'_i, r'_i)$ do not overlap. Therefore, $r^*_{i-1} < l'_i$ holds, so the intervals do not overlap. Also, the number of chosen intervals is $k$, so it is an optimal solution.
\end{proof}

\begin{algorithm}
    \caption{$\textsc{IntervalEF1}(G = (M, E), v)$}
    \begin{algorithmic}[1]
        \State $[l_1, r_1), \dots, [l_m, r_m) \gets$ intervals corresponding to each vertex in $G$ such that $l_1, r_1, \dots, l_m, r_m$ are distinct
        \State $Z \gets$ an optimal solution of \textsc{IntervalScheduling} with $c = 2$ \Comment{Let $Z = \{z_1, \dots, z_k\}$ to be $r_{z_1} < \dots < r_{z_k}$}
        \State $Z_1 \gets \{z_1, z_3, z_5, \dots\}$
        \State $Z_2 \gets \{z_2, z_4, z_6, \dots\}$
        \If{$v(Z_1) < v(Z_2)$} swap $Z_1, Z_2$ \EndIf
        \For{$z \in Z_2$}
            \If{intervals contained in $Z_1 \cup \{z\}$ do not overlap}
                \State Add $z$ to $Z_1$ and remove $z$ from $Z_2$
            \EndIf
        \EndFor
        \State $X'_1, X'_2 \gets$ a ``greedy solution'' and a ``greedy solution from the opposite direction'' for \textsc{Interval Scheduling} of intervals in $M \setminus Z_1$ with $c = 1$
        \State $\mathcal{C}_1 \gets$ chain of allocations generated by Algorithm \ref{alg:ef1-chain} setting $S = Z_1, X_1 = X'_1, X_2 = X'_2$
        \State $\mathcal{C}_2 \gets$ a chain of allocations which starts with $(Z_1, Z_2)$ and ends with $(Z_1, X'_2)$
        \State $\mathcal{C}_3 \gets$ a chain of allocations which starts with $(X'_1, Z_1)$ and ends with $(Z_2, Z_1)$
        \State $\mathcal{C} \gets$ chain of allocations created by concatenating $\mathcal{C}_2, \mathcal{C}_1, \mathcal{C}_3$ in this order \Comment{$\mathcal{C}$ is a gapless chain}
        \State \Return any maximal EF1 allocation in $\mathcal{C}$
    \end{algorithmic}
\end{algorithm}

\begin{proof}[Proof of Theorem~\ref{thm:twoagents_interval}]
    We show that \textsc{IntervalEF1} returns a maximal EF1 allocation. First, the resulting $Z_1, Z_2$ meet the following conditions:
    \begin{itemize}
        \item $Z_1$ is maximal in $M$. Due to lines 6-8, any $z \in Z_2$ cannot be added to $Z_1$. Also, if $z \in Z \setminus (Z_1 \cup Z_2)$ can be added, this contradicts the fact that $Z$ is an optimal solution for \textsc{IntervalScheduling} with $c = 2$.
        \item $v(Z_1) \geq v(Z_2)$ due to line 5 and monotonicity.
    \end{itemize}
    
    Next, we consider $\mathcal{C}_1$. In Algorithm \ref{alg:ef1-chain}, when $S$ is sorted, i.e. $l_{s_1} < r_{s_1} < \dots < l_{s_k} < r_{s_k}$, it suffices to regard $p_t$ as $l_t$ and $q_t$ as $r_t$.\footnote{Note that, for example, the actual $p_t$ can take the same value even if $l_t$'s are different, but it does not affect the construction of $X_2$ because ties of $p_t$ can be resolved in any way.} Then, $X_1$ becomes a greedy solution, and $X_2$ becomes a greedy solution from opposite direction. Hence, the construction of $\mathcal{C}_1$ is justified.

    Finally, we consider $\mathcal{C}_3$. We use Lemma \ref{lem:interval_lemma2} to construct this chain. We note that both $X'_1$ and $Z_2$ are optimal solutions for \textsc{IntervalScheduling} of intervals in $M \setminus Z_1$.
    \begin{itemize}
        \item $X'_1$ is an optimal solution because the greedy solution is optimal (Lemma \ref{lem:interval_lemma1}).
        \item $Z_2$ is an optimal solution because otherwise, this would contradict the optimality of $Z$ (for \textsc{IntervalScheduling} with $c = 2$).
    \end{itemize}
    We use Lemma \ref{lem:interval_lemma2} by setting $[l'_1, r'_1), \dots, [l'_k, r'_k)$ to intervals of $Z_2$. We create $\mathcal{C}_3$ by arranging the following allocation for $i = k, \dots, 0$ in this order.
    \begin{equation*}
        (\{[l^*_1, r^*_1), \dots, [l^*_{i-1}, r^*_{i-1}), [l'_i, r'_i), \dots, [l'_k, r'_k)\}, Z_1)
    \end{equation*}
    They are maximal due to Lemma \ref{lem:interval_lemma2}. By symmetry, we can construct $\mathcal{C}_2$ in a similar way.

    Combining all these facts, $\mathcal{C}$ is a gapless chain. Note that $\mathcal{C}$ starts with $(Z_1, Z_2)$ and ends with $(Z_2, Z_1)$, where $v(Z_1) \geq v(Z_2)$ is ensured.

    Finally, we analyze the running time of the algorithm.
    \begin{itemize}
        \item Reconstructing $[l_1, r_1), \dots, [l_m, r_m)$ can be done in $O(m)$ time, due to \cite{Hsu93}.
        \item Obtaining an optimal solution of \textsc{IntervalScheduling} can be done in $O(m)$ time by the greedy algorithm.
        \item Making $Z_1$ maximal (lines 6-8) can be done in $O(m \log m)$ time by using binary search tree.
        \item The length of $\mathcal{C}$ is at most $3m+1$, so it takes $O(m \cdot T(m))$ time to find a maximal EF1 allocation among them.
    \end{itemize}
    Thus, the algorithm runs in $O(m \log m + m \cdot T(m))$ time.
\end{proof}

%\section{Omitted Material from Section~\ref{sec:three}}

%\ExampleFourAgents*

\section{Uniform Valuations}\label{sec:uniform}

We consider the case of uniform valuations, i.e., $v_i(S) = \left|S\right|$ for all $S \subseteq M$. For such valuations, our problem is closely related to an \emph{equitable coloring} of a graph, where the number of vertices assigned to each color differs by at most $1$. The key difference is that an  \emph{equitable coloring} must be complete---no vertex can remain uncolored---whereas in our problem, we aim to find a \emph{maximal equitable (partial) coloring} defined as follows.  

% \pasin{Add discuss here on how it's different from the usual equitable coloring?} \pasin{Also, maybe we should add ``partial'' before ``coloring''?}

Given a graph $G = (V, E)$ and $n \in \N$, a \emph{maximal equitable $n$-coloring} is a subpartition $(S_1, \dots, S_n)$ of $V$, where $S_1, \dots, S_n \subseteq V$, with the following conditions:
    \begin{itemize}
        \item $S_1, \dots, S_n$ are disjoint.
        \item (independence) For any pair of distinct vertices $x, y \in S_i$, the vertices $x, y$ are not adjacent.
        \item (maximality) For every $v \in V \setminus (S_1 \cup \dots \cup S_n)$ and $i \in \{1, \dots, n\}$, $v$ has an adjacent vertex in $S_i$. 
        \item (equitability) $\max_{i \in [n]} \left|S_i\right| - \min_{i \in [n]} \left|S_i\right| \leq 1$.
    \end{itemize}
%Here, vertices in $S_i$ are colored in ``color $i$", and vertices not in $S_1, \dots, S_q$ are ``non-colored vertices".

Below, we prove that a maximal equitable $n$-coloring exists when $G$ is a tree. It remains an open problem whether this can be extended to all graphs.

\begin{theorem}
For every tree $T = (V, E)$ and every $n \in \N$, there exists a maximal equitable $n$-coloring of $T$.
    \label{thm:uni-tree2}
\end{theorem}

\begin{proof}
    Consider $T$ as a rooted tree with root $r$. We prove a stronger fact that there exists a maximal equitable $n$-coloring of $T$ such that $r$ is colored with a \emph{higher color} or is uncolored. Here, for a coloring $\mathcal{S} = (S_1, \dots, S_n)$, we define that color $i$ is a \emph{higher color} if $|S_i| = \max_{j \in [n]} |S_j|$.

    We will prove this by strong induction on the number of vertices of $T 
= (V, E)$. The base case $|V| = 1$ is trivial.

    For the inductive step, consider any tree $T 
= (V, E)$ with $|V| > 1$ and suppose that %a maximal equitable $n$-coloring exists 
the statement holds
for all trees with smaller number of vertices.
    Let $r_1, \dots, r_k$ be the children of $r$, and let $T_1, \dots, T_k$ be the subtree of $r_1, \dots, r_k$, respectively. By the inductive hypothesis, %we assume that
    for each $i$, there exists a maximal equitable $n$-coloring of $T_i$ where $r_i$ is a higher color or is uncolored; let it be $\mathcal{S}^{(i)} = (S^{(i)}_1, \dots, S^{(i)}_n)$. Without loss of generality, we assume $|S^{(i)}_1| \geq \dots \geq |S^{(i)}_n|$. We say that $T_i$ is a \emph{singular subtree} if $\mathcal{S}^{(i)}$ has only one higher color (color $1$) and $r_i$ has color $1$. Without loss of generality, for some $c$, $T_1, \dots, T_c$ are singular subtrees and $T_{c+1}, \dots, T_k$ are not.
    
    First, we aim to obtain an equitable $n$-coloring $\mathcal{S}$ of $T_1 \cup \dots \cup T_k$ in the following way:
    \begin{enumerate}
        \item Initialize $\mathcal{S} = (\emptyset, \dots, \emptyset)$ and $x = 0$.
        \item For $i = 1, \dots, k$, do the following:
        \begin{enumerate}
            \item $(S_1, \dots, S_n) \gets (S_1 \cup S^{(i)}_{n-x+1}, \dots, S_x \cup S^{(i)}_n, S_{x+1} \cup S^{(i)}_1, \dots, S_n \cup S^{(i)}_{n-x})$ where the indices of $\mathcal{S}^{(i)}$ wrap around to 1 after $n$, and,
            \item $x \gets (x + (\text{number of higher colors in $\mathcal{S}^{(i)}$})) \bmod n$.
        \end{enumerate}
    \end{enumerate}
    It is not difficult to see that, for every loop just after step 2 (b), colors $1, \dots, x$ are higher colors of $\mathcal{S}$, and colors $x+1, \dots, n$ are not (or, when $x = 0$, all colors are higher colors). Also, observe that the algorithm works in a way that $r_1, \dots, r_{\min(c, n)}$ will have colors $1, \cdots, \min(c, n)$, respectively, as $T_1, \dots, T_{\min(c, n)}$ are singular subtrees.
    
    The problem is that, in order to obtain a coloring of $T$, we need to consider vertex $r$. We divide it into two cases to solve this issue.
    \begin{itemize}
        \item (Case 1. $c \geq n$) $r_1, \dots, r_n$ will have colors $1, \dots, n$, respectively. Therefore, $r$ can be left uncolored, and $\mathcal{S}$ is already a maximum equitable $n$-coloring of $T$.
        \item (Case 2. $c < n$) We aim to color vertex $r$ with color $n$. The case when this cannot be done is that some $r_i$ already has color $n$. However, since $r_1, \dots, r_c$ have colors $1, \dots, c \ (< n)$, we know that $T_i$ is not a singular subtree. So, there is another higher color inside $T_i$, say color $x$ (in the context of coloring $\mathcal{S}$; $x \neq n$ must hold). Then, we can swap color $x$ and color $n$ inside the entire $T_i$ --- then $r_i$ no longer has color $n$, and the number of vertices of each color does not change. We repeat this process until no $r_i$ has color $n$. Since color $n$ is used by the least number of vertices, coloring vertex $r$ in color $n$ does not break the equitability condition, and color $n$ becomes a higher color.
    \end{itemize}
    
    Therefore, we obtained the desired coloring of $T$.
\end{proof}

\end{document}